%% file: arxiv-speed-robust.tex
\documentclass[a4paper,USenglish,numberwithinsect]{lipics-v2021}

\hideLIPIcs
\nolinenumbers

\title{Speed-robust scheduling revisited}

\author{Josef Mina{\v{r}}{\'\i}k}
{Computer Science Institute of Charles Univ.,
Faculty of Mathematics and Physics,
Prague, Czechia}
{minarjos00@gmail.com}{}
{}
    
\author{Ji{\v{r}}{\'\i} Sgall}
{Computer Science Institute of Charles Univ.,
Faculty of Mathematics and Physics,
Prague, Czechia}
{sgall@iuuk.mff.cuni.cz}
{https://orcid.org/0000-0003-3658-4848}
{}

\authorrunning{J. Mina{\v{r}}{\'\i}k and J. Sgall} 

\Copyright{Josef Mina{\v{r}}{\'\i}k and Ji{\v{r}}{\'\i} Sgall} 



\funding{Partially supported by grant 24-10306S of GA~\v{C}R.} 

\counterwithin{figure}{section}

\EventEditors{Amit Kumar and Noga Ron-Zewi}
\EventNoEds{2}
\EventLongTitle{Approximation, Randomization, and Combinatorial Optimization. Algorithms and Techniques (APPROX/RANDOM 2024)}
\EventShortTitle{\mbox{\scriptsize APPROX/RANDOM 2024}}
\EventAcronym{APPROX/RANDOM}
\EventYear{2024}
\EventDate{August 28--30, 2024}
\EventLocation{London School of Economics, London, UK}
\EventLogo{}
\SeriesVolume{317}
\ArticleNo{8}
\usepackage[utf8]{inputenc}

\usepackage{fancyvrb}       

\usepackage{algorithm}
\usepackage[noend]{algpseudocode}

\usepackage{pgfplots}
\pgfplotsset{width=8cm,compat=1.9}

\DefineVerbatimEnvironment{code}{Verbatim}{fontsize=\small, frame=single}

\DeclareMathOperator{\cost}{cost}
\DeclareMathOperator{\size}{size}

\numberwithin{equation}{section}

\begin{document}

\maketitle

\begin{abstract}
Speed-robust scheduling is the following two-stage problem of scheduling $n$ jobs on $m$ uniformly related machines.
In the first stage, the algorithm receives the value of $m$ and the processing times of $n$ jobs; it has to partition the jobs into $b$ groups called bags. 
In the second stage, the machine speeds are revealed and the bags are assigned to the machines, i.e., the algorithm produces a schedule where all the jobs in the same bag are assigned to the same machine.
The objective is to minimize the makespan (the length of the schedule). 
The algorithm is compared to the optimal schedule and it is called $\rho$-robust, if its makespan is always at most $\rho$ times the optimal one.

Our main result is an improved bound for equal-size jobs for $b=m$. We give an upper bound of $1.6$. This improves previous bound of $1.8$ and it is almost tight in the light of previous lower bound of $1.58$. 
Second, for infinitesimally small jobs, we give tight upper and lower bounds for the
case when $b\geq m$. This generalizes and simplifies the previous bounds for $b=m$.
Finally, we introduce a new special case with relatively small jobs for which we give an algorithm whose robustness is close to that of infinitesimal jobs and thus gives better than $2$-robust for a large class of inputs. 
\end{abstract}

\section{Introduction}

Speed-robust scheduling is a two-stage problem that was introduced by Eberle \textit{et al.}~\cite{eberle}. The eventual goal is to schedule on $m$ uniformly related machines, however their speeds are not known at the beginning.
In the first stage, the algorithm receives the value of $m$ and the processing times of $n$ jobs; it has to partition the jobs into $b$ groups called bags. 
In the second stage, the machine speeds are revealed and the bags are assigned to the machines, i.e., the algorithm produces a schedule where all the jobs in the same bag are assigned to the same machine.
The objective is to minimize the makespan (the length of the schedule). 
The algorithm is compared to the optimal schedule of the jobs on the
machines with known speeds;
it is called $\rho$-robust,
if its makespan is always at most $\rho$ times the optimal one.

This problem is motivated by situations like the following one. 
Suppose that you have $n$ computational tasks that you want to solve.
You have a computational cluster available, but with unknown parameters.
You only know that there will be (at most) $m$ machines available on the cluster.
You do not know anything about the performance of the machines---some of the machines might be faster than others; you only know that there will be (at most) $m$ machines available on the cluster.
Furthermore, you can submit at most $b$ different tasks to the cluster.
Hence you will have to partition your $n$ tasks into at most $b$ groups.
One such group will then have to be executed on one machine.
The cluster will then schedule the groups optimally, knowing the speeds of the machines, and minimize the makespan.

Studying uncertainty in scheduling has a long history. In the
classical online scheduling~\cite{online-survey}, the machine
environment is usually fixed and the uncertainty stems from job
arrivals.  Considering uncertainty in the machine environment is less
frequent.  One early example is the work of Csan\'ad and
Noga~\cite{Noga-speeds}, where additional machines can be bought for a
certain cost.  A substantial body of research with changing machine
speeds is the area of dynamic speed scaling, in particular in the
context of minimizing the power consumption,
see~\cite{albers-survey,pruhs-survey}; however, note that here the
changing speeds are not a part of the adversarial environment but used
by the algorithm to its advantage.  Another direction considers online
scheduling with unavailability
periods~\cite{unavailability-survey}. One-machine scheduling with
adversarially changing machine speed was considered
in~\cite{epstein-speeds} in the context of unreliable machines.

Completely reversing the scenario with all jobs known from the
beginning but uncertain machine environment is a recent new model
introduced by Stein and Zhong~\cite{zero_one} and Eberle {\it et
  al.}~\cite{eberle}, see Section~\ref{sec-related}.


\subsection{Formal definitions}

Formally, in the first stage, we receive three positive integers $n, m, b$  and $n$ non-negative real numbers $p_1, \ldots, p_n$ representing the processing times of $n$ jobs. The total processing time is denoted $P=\sum_{j=1}^m p_j$.
The output of our first-stage algorithm is a mapping $B: \{1, \ldots, n\} \to \{1, \ldots, b\}$, where $B(j) = i$ represents the fact that the job $j$ was assigned to the bag $i$.
The sum of the processing times of all the jobs assigned to bag $i$ the \textit{size} of bag $i$ and denoted 
$a_i = \sum_{j:B(j) = i} p_j$.
The exact mapping $B$ is not important for the second stage since the makespan depends only on the bag sizes.

In the second stage, we are given the bag sizes $a_1,\ldots,a_b$ and the previously unknown machine speeds $s_1, \ldots, s_m\geq 0$, not all equal to $0$.
We partition the bag indices $\{1,\ldots,b\}$ into $m$ sets $M_1$, \ldots, $M_m$, representing the assignment to the $m$ machines.
Machine $i$ then has a completion time $C_i=(\sum_{j\in M_i}a_j)/s_i$; for $s_i=0$ we require $M_i=\emptyset$ and set $C_i=0$, i.e., machine of speed $0$ does not accept any jobs. Finally, $C_{\max}=\max_{i=1}^m C_i$ is the makespan, i.e., the length of the schedule.

Let $C_{\max}^*$ denote the makespan of the adversary, who does not have to create bags and can assign jobs directly to machines. Alternatively and equivalently, 
the adversary also creates bags, but with the knowledge of the speeds already in the first stage.

We call a first stage algorithm $\rho$-\textit{robust} if, for
all possible inputs and for all possible choices of machine speeds,
there exists a second-stage assignment of bags to machines such that
$C_{\max} \le \rho\cdot C_{\max}^*$. 
Intuitively, an algorithm is $\rho$-robust if it performs at most
$\rho$ times worse than the adversary. 

The previous definition implicitly assumes that the second stage is
solved optimally.  This is reasonable, as the scheduling on uniformly
related machines allows PTAS, see~\cite{speeds,textbook}, so the
chosen (presumably optimal) second-stage solution can be replaced by
an arbitrarily good approximation.  Also, our proofs show that
the second-stage algorithm can be implemented by efficient greedy
algorithms without any loss of performance, once the optimal makespan
or its approximation is known.

We call the special cases of the problem {\it sand\/}, {\it bricks\/}, {\it rocks\/} and {\it pebbles\/}.
Sand, bricks, and rocks were introduced by Eberle {\it et al.}~\cite{eberle}. 
These words represent the types of jobs. 
\begin{itemize}
    \item 
        Rocks can be any shape or size and represent jobs of arbitrary processing time. This is the most general setting.
    \item
        Bricks are all the same and represent jobs with equal processing times.
    \item
        Sand grains are very small and represent infinitesimally small processing times.
    \item
        Pebbles represent jobs that are relatively small compared to the average load of all machines.
        We call an instance of speed-robust scheduling $q$-pebbles if
        $p_j \le q \cdot \frac{P}{m}$
        holds for all jobs $j$.
\end{itemize}

\subsection{Previous results}
\label{sec-related}

The two-stage scheduling problem with uncertainty in the machine
environment was introduced by Stein and Zhong
\cite{zero_one}.  They focused on the case of $m$
identical machines where in the second stage some machine might fail
and then do not process any tasks. This amounts to a special case of
speed-robust scheduling where $s_i \in \{0, 1\}$ for $1 \le i \le
m$. They gave lower bounds of $4/3$ for equal-size jobs (bricks) and
$(\sqrt 2+1)/2\approx 1.207$ for infinitesimal jobs (sand). Their
algorithms were later improved by Eberle \textit{et
  al.}~\cite{eberle} to algorithms matching the lower
bounds in both cases.

Our immediate predecessor, Eberle \textit{et al.}~\cite{eberle},
introduced the speed-robust scheduling for general speeds, i.e., on
uniformly related machines.  They studied mainly the case $b=m$, i.e.,
the case when the number of bags is equal to the number of machines.
For this case they gave tight bounds for sand for every $m$, for large
$m$ the bound approaches $e/(e-1)\approx 1.58$. For equal-size jobs
(bricks), they have shown an upper bound of $1.8$.

For the most general case of rocks, the strongest known result is the
algorithm with the robustness factor at most $1 + (m-1)/b$, which
equals $2-1/m$ for $b=m$, given also by Eberle \textit{et
  al.}~\cite{eberle}.
It remains an interesting open problem to improve this bound, in
particular to give an upper bound $2-\varepsilon$ for rocks and $b=m$.

\subsection{Our results}

We now describe our results and compare them to the previous ones in
each of the scenarios.

\subparagraph*{Sand.}
For sand, we give matching lower and upper bounds for any $b$ and $m$. Namely, for $b\geq m$ we give an optimal algorithm which is $\overline{\rho}(m, b)$-robust for  
\begin{equation}
\label{eq-rho}
    \overline{\rho}(m, b) = \frac{m^b}{m^b - (m-1)^b}
    =\frac1{1-\left(1-\frac1m\right)^b}\,.    
\end{equation}
This matches the results of Eberle \textit{et
  al.}~\cite{eberle} who gave an algorithm with the
robustness factor
equal to $\overline{\rho}(m,b)\leq e/(e-1)\approx 1.58$ for $b=m$,
generalizes them to arbitrary $b\geq m$ and
significantly simplifies the proof.

An interesting case is when the
number of bags is a constant multiple of $m$.  For a fixed $\alpha\geq
1$ and $b=\alpha m$,
our bound approaches $1/(1-e^{-\alpha})$ from below for a large $m$.
For example,
doubling the number of bags to $b=2m$ decreases the robustness factor
from $1.58$ to $1.16$.

If $b<m$, the second-stage algorithm uses only the $b$ fastest machines,
so we can decrease $m$ to $m'=b$ and tight results with robustness factor
$\overline{\rho}(m',b)=\overline{\rho}(b,b)$ follow already from~\cite{eberle}.

\subparagraph*{Pebbles.}
For the new case of $q$-pebbles and $b\geq m$, we give a
$(\overline{\rho}(m,b) + q)$-robust algorithm.  For $p<0.42$, this
gives an algorithm with the robustness factor below $2$, i.e., below the currently strongest known
upper bound for rocks.

\subparagraph*{Bricks.}
As our main result, we give a $1.6$-robust algorithm for bricks for $b=m$. This improves the bound of $1.8$ from Eberle \textit{et al.}~\cite{eberle}.

Furthermore, as a direct application of our results for pebbles we give a $(\overline{\rho}(m, b) + m/n)$-robust algorithm for any $n$ and $b \ge m$. 
This improves and generalizes the
$((1 + m/n)\overline{\rho}(m,m))$-robust algorithm for $b=m$ given by Eberle \textit{et al.}~\cite{eberle}. Namely, we improve the 
multiplicative factor of $(1 + m/n)$ to only an additive term of $m/n$. 

\subparagraph*{Structure of the paper.}
We give some general preliminaries in Section~\ref{sec-prelim}. We
give the results for sand and pebbles in Sections~\ref{sec-sand}
and~\ref{sec-pebbles}.  We focus on our main result for bricks in
Section~\ref{sec-bricks}.
Some small cases need computer verification
or tabulation of parameters, results of these are given in
Appendix~\ref{sec-bricks-computer}.

\section{Preliminaries}
\label{sec-prelim}

We assume that the processing times, the machine speeds, and the
bag sizes are always listed in a non-increasing order.

In the rest of this paper, we will make two assumptions below that
restrict the speeds to particular special cases. This is without loss of
generality, leveraging the fact that the algorithm must commit the bag
sizes in the first phase without knowing the speeds.

\begin{itemize}
    \item 
The optimal makespan is equal to 1. This implies that the robustness
factor is equal to the makespan of the algorithm.
        
Scaling all the speeds does not change the ratio of the makespans of
our algorithm and the adversary.  Thus for every instance of the
problem, there exists another instance with $C^*_{\max} = 1$ that
differs only in the speeds and the ratio of makespans of our algorithm
and the adversary remains the same.  It follows that any first-stage
algorithm that is $\rho$-robust for instances with $C^*_{\max} = 1$ is
$\rho$-robust for general instances, too.
    \item 
        The sum of the processing times of all jobs equals to the sum of the speeds of all the machines, i.e., $P=\sum_{i=1}^m p_i = \sum_{i=1}^m s_i$.
        In other words, the adversary is fully utilizing all the machines, and the completion time of all the machines with non-zero speed is equal to 1, using the previous assumption.
        
        If there is some machine $i$ with $s_i>0$ and completion time $C^*_i < 1$ in the optimal schedule, we change its speed to $s'_i = C^*_is_i$.
        This does not change the optimal makespan of the adversary and the makespan of the algorithm can only increase.
Once again, it follows that any first-stage algorithm that is $\rho$-robust
for these special instances is also $\rho$-robust for general instances. 
\end{itemize}

For the second stage, typically, we use a simple greedy algorithm for
the second stage instead of analyzing the optimal schedule.
Technically, for an algorithm we need to know the optimal makespan (to
modify the speeds appropriately, according to the assumptions
above). However note that first we can approximate the makespan and
second the algorithm is only used as a tool in the analysis.

For sand and pebbles we use Algorithm \nameref{alg-greedy} (see
below), a variant of the well-known LPT algorithm.  It is
parameterized by $\rho$, the robustness factor to be achieved.  At the
beginning, every machine is assigned a capacity equal to its speed
multiplied by $\rho$.  The algorithm then goes through all the bags
from large to small, assigns them on the machine with the largest
capacity remaining, and decreases the capacity appropriately.  If the
capacities remain non-negative at the end, the makespan of the created
assignment is at most $\rho$ since machine $i$ has been assigned jobs
of total processing times at most $\rho s_i$.

\begin{algorithm}
    \caption{\normalfont \scshape GreedyAssignment}
    \begin{algorithmic}
        \State {\bf Input}: bag sizes $a_1 \ge \cdots \ge a_b$;
         machine speeds $s_1\ge \cdots \ge s_m$;
        desired robustness factor $\rho$
        \For{$i \gets 1$ to $m$}
            \State 
            $c_i \gets \rho s_i$
            \Comment{Initialize the capacities of all machines}
            \State 
            $M_i=\emptyset$
             \Comment{Initialize the assignment}
        \EndFor
        \For{$k \gets 1$ to $b$}
            \State $i \gets $ index of a machine with the largest $c_i$
            \State $M_i \gets M_i\cup \{k\}$
            \Comment{Assign bag $k$ to machine $i$}
            \State $c_i \gets c_i - a_k$
            \Comment{Decrease the remaining capacity of the selected machine}
        \EndFor
        \State \Return $M_1, \ldots, M_m$
    \end{algorithmic}
    \label{alg-greedy}
\end{algorithm}

We use this to formulate the following sufficient condition for
$\rho$-robustness of an algorithm
which is instrumental in proving the upper bounds for sand and pebbles. 

\begin{theorem}
    \label{thm-greedy}
 If a first-stage algorithm always produces bag sizes satisfying inequalities
    $\displaystyle
        a_k \le \frac{\rho P - \sum_{j=1}^{k-1}a_j}{m}\,,
    $
    for all $k=1,\ldots,b$, then the algorithm is $\rho$-robust.    
\end{theorem}
\begin{proof}
    Recall that we assume $\sum_{i=1}^m s_i = P$.
    We claim that the second-stage algorithm \nameref{alg-greedy} produces an assignment with makespan at most $\rho$.
    We only need to show that there is a machine with capacity at least $a_k$ when assigning the $k$th bag.
    The initial total capacity was $\rho P$ and was already decreased by $\sum_{j=1}^{k-1}a_j$ at the time of assigning bag $a_k$.
    It follows that the remaining capacity is equal to
    $\rho P - \sum_{j= 1}^{k - 1}a_j$  
    and thus there exists a machine with capacity at least
    $(\rho P - \sum_{j=1}^{k-1}a_j)/m\geq a_k$.
\end{proof}

\input{arxiv-sand-pebbles}

\input{arxiv-bricks}

\section*{Conclusions}

Our main result still leaves a small gap in the bounds for bricks
(equal-length jobs) and $b=m$ between the lower bound of
$e/(e-1)\approx 1.58$ and our upper bound of $1.6$. Our algorithm
\nameref{alg-bag-size} does not admit a smaller robustness factor than
$1.6$, as is shown for $n = 45$ and $m = 9$ in
Figure~\ref{fig_45_9}. So a smaller upper bound would need some
additional techniques or special handling of some cases.  Eberle
\textit{et al.}~\cite{eberle} give an example that shows a
lower bound for bricks that is larger than $\overline{\rho}(m, m)$ for
$m=6$. Although the value of the bound is below the limit value
$e/(e-1)$, this may be taken as a weak evidence that matching the
lower bound may be hard.

The main open problem in this model remains to find a 
$(2-\varepsilon)$-robust algorithm for the general case and $b=m$. 

\subparagraph*{Acknowledgments}
We are grateful to Franziska Eberle for insightful discussions and
  comments, and to the anonymous referees for detailed and useful reviews.

\bibliographystyle{abbrv}
\bibliography{arxiv-speed-robust}

\appendix

\input{arxiv-bricks-computer}

\end{document}

%% file: arxiv-sand-pebbles.tex

\section{Sand}
\label{sec-sand}


Intuitively, the case of sand corresponds to the limit case where $n$
is large and all the jobs are small and have equal sizes. One can view
this as an infinite number of infinitesimal jobs.

More formally, we are given just $m$, $b$, and $P$ as an input of the
first stage.  The result of the first stage are $b$ non-negative reals
$a_1, \ldots, a_b$ whose sum equals $P$. The formulation of the second
stage remains the same.

The model of infinitesimally small jobs resembles preemptive scheduling. 
In the optimal algorithm for preemptive scheduling~\cite{preemptive}, one needs to maintain the loads of machines in a geometric sequence with common ratio $m/(m-1)$ for $m$ machines, roughly speaking.
The proofs for sand show that here the same geometric sequence is also crucial, in particular it is used for the bag sizes in the algorithm.
We now describe the sequence and state its properties useful both for the upper and lower bounds. 

We define $U = m^b$, $L = m^b - (m - 1)^b$ and $t_j = m^{b - j}(m -
1)^{j - 1}$ for $j \in \{1, \ldots, b\}$.

Observe that equation~(\ref{eq-rho}) defines $\overline{\rho}$ as  $\overline{\rho}(m,b) = U/L$. 
\begin{lemma}
   For all $k=1,\ldots,b$, it holds that $
        \sum_{j = 1}^k t_j = U - (m - 1)t_{k}$. 
        In particular, $\sum_{j = 1}^b t_j=U-(m-1)^b=L$.
    \label{lem-prefix-sums}
\end{lemma}
\begin{proof}
    We proceed by induction on $k$.
    The lemma holds for $k = 1$ since $U = mt_1$ and thus $t_1 = U - (m - 1)t_1$.
    Now suppose it holds for $k$. We can derive
    \[
        \sum_{j = 1}^{k + 1} t_j = U - (m - 1)t_k + t_{k + 1} = U - m^{b - k}(m - 1)^{k} + m^{b - k - 1}(m - 1)^k = U - (m - 1)t_{k + 1}
    \]
    which completes the induction step.
\end{proof}

\begin{figure}[ht]
    \centering
    \includegraphics{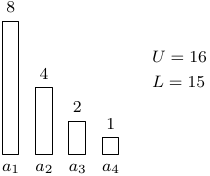}
    \caption{An example of bag sizes chosen for $m = 2$ and $b = 4$.}
    \label{fig-sand_bagsize}
\end{figure}

To get some intuition behind the algorithm for sand, it might be
useful to consider the case $m = 2$, see
Figure~\ref{fig-sand_bagsize}.  Suppose that $P = L = 2^b-1$, choose
the bag sizes $a_k = t_k = 2^{b - k}$.  For $m=2$ the sizes are powers
of two, so it is easy to see that we can achieve the robustness ratio
of $1 + 1/(2^b - 1)= 2^b/(2^b - 1)$ as follows: The adversary chooses
any speeds $s_1, s_2$ such that $s_1 + s_2 =P= 2^b - 1$.  The
capacities of the machines (as in \nameref{alg-greedy}) then satisfy
$c_1 + c_2 = 2^b$ and thus $\lfloor c_1 \rfloor + \lfloor c_2 \rfloor
\ge 2^b - 1$.  We can express $\lfloor c_1 \rfloor$ in binary, assign
the corresponding bags on the first machine and the remaining bags to
the second machine.

\subsection{Upper bound}
\label{sec-sand-upper}

We use a different approach than Eberle {\it et al.}~\cite{eberle} for the proof of the upper bound.
We choose the same bag sizes (for $b=m$) but we simplify the proof by use of 
Theorem~\ref{thm-greedy}.
Algorithm \nameref{alg-sand-optimal} describes the bag sizes.
Note that the sum of bag sizes produced by \nameref{alg-sand-optimal} is $P$, using Lemma~\ref{lem-prefix-sums}.

\begin{algorithm}
    \caption{\normalfont \scshape Sand}
    \begin{algorithmic}
        \State {\bf Input}: number of bags $b$;
        number of machines $m$;
        total amount of sand $P$
        \State $L \gets m^b - (m - 1)^b$
        \For{$j \gets 1$ to $b$}
            $a_j \gets t_j \frac{P}{L}$
        \EndFor
        \State 
        \Return $a_1, a_2, \ldots, a_b$
    \end{algorithmic}
    \label{alg-sand-optimal}
\end{algorithm}

\begin{theorem}
    Algorithm \nameref{alg-sand-optimal} is $\overline{\rho}(m,b)$-robust for sand, for $\overline{\rho}$ defined by {\rm (\ref{eq-rho})}.
    \label{thm-sand-upper}
\end{theorem}

\begin{proof}
    We assume $P = L$ since it does not change the ratio of our makespan and the makespan of the adversary.
    Under this assumption, \nameref{alg-sand-optimal} produces bag sizes $a_k = t_k$.
    
    It is sufficient to show that the bag sizes produced by \nameref{alg-sand-optimal} satisfy the condition of Theorem~\ref{thm-greedy}.
    Let us prove the $k$th inequality in the assumption of the theorem. We have
    \[
        \overline{\rho}(m,b) P - \sum_{j = 1}^{k - 1}a_j = \frac{U}{L} L - \sum_{j = 1}^{k-1}t_j = U - \sum_{j = 1}^{k}t_j + t_k\,.
    \]
    According to Lemma \ref{lem-prefix-sums}, we can simplify the
    right-hand side as follows.
    \[
        U - \sum_{j = 1}^{k}t_j + t_k = U - (U - (m-1) t_k) + t_k = mt_k = ma_k\,.
        \]
The $k$th inequality in the assumption of Theorem~\ref{thm-greedy}
follows, in fact it holds with equality.
        
Theorem~\ref{thm-greedy} now implies that there exists an assignment
with makespan at most $\overline{\rho}(m,b)$. 
\end{proof}

\subsection{Lower bound}
\label{sec-sand-lower}

The following proof is a slightly modified and generalized version of the proof by Eberle {\it et al.}~\cite{eberle}.
The main difference is that we do not require the number of bags and machines to be the same.

\begin{theorem}
    No deterministic algorithm for sand may have a robustness factor smaller than $\overline{\rho}(m,b)$, for $\overline{\rho}$ defined by {\rm (\ref{eq-rho})}.
    \label{thm-sand-lower}
\end{theorem}
\begin{proof}
    Let us without loss of generality assume $P = U$ (be aware that we assumed $P = L$ in the proof of the upper bound).
    Let us denote the chosen bag sizes by $a_1 \ge \cdots \ge a_b$.
    We will restrict the adversary to $b$ different speed configurations indexed by $k$, where
    \[
        \mathcal{S}_k = \{s_1 = U - (m-1) t_k,~ s_2 = t_k,~s_3 = t_k,~\ldots,~s_{m}=t_k\}\,.
    \]
See Figure~\ref{fig-sand_speeds} for an example. Note that the sum of machine speeds is equal to $U$ in every configuration and hence the makespan of the adversary is indeed 1 as we always assume. 
    In every speed configuration, there are $m - 1$ slow machines and one fast machine, since 
    \[
        s_1 = U - (m - 1)t_k = \sum_{j=1}^kt_j \ge t_k\,.
    \]

        \begin{figure}[ht]
        \centering
        \includegraphics{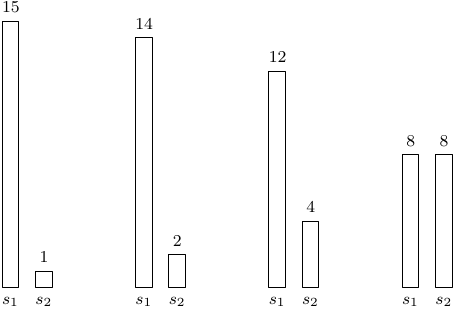}
        \caption{An example of speed configurations considered by the adversary for $m = 2$ and $b = 4$.}
        \label{fig-sand_speeds}
    \end{figure}

    Let $k_{\max}$ be the largest index such that $a_{k_{\max}} \ge \frac{U}{L} t_{k_{\max}}$.
    This index must exist since
    \[
        \sum_{j=1}^ba_j = U = \frac UL L = \frac UL \sum_{j=1}^bt_j\,.
    \]
    Now let the adversary choose the speed configuration $\mathcal{S}_{k_{\max}}$. We distinguish two cases depending on the bag assignment in the second stage.

\noindent \textbf{Case 1:}     
At least one of the bags $a_1, \ldots, a_{k_{\max}}$ is assigned to a slow machine. 
            The makespan is at least the completion time of this
            machine which is at least 
            \[
                \frac{a_j}{t_{k_{\max}}} \ge \frac{a_{k_{\max}}}{t_{k_{\max}}} \ge \frac UL\,.
            \] 

\noindent \textbf{Case 2:}     
            All of the bags $a_1, \ldots, a_{k_{\max}}$ are assigned to the fast machine. Total size of the bags assigned to the fast machine is at least
            \[
                 \sum_{j=1}^{k_{\max}}a_j = U - \sum_{j=k_{\max} + 1}^{b}a_j\,.
            \]
            By definition of $k_{\max}$ it holds that $a_j < \frac{U}{L}t_j$ for $j > k_{\max}$ and we can bound
            \[
               U-\sum_{j=k_{\max} + 1}^{b}a_j \ge U - \frac UL\sum_{j=k_{\max} + 1}^{b}t_j\,.
            \] 
            Since $\sum_{j = 1}^b t_j = L$, we can rearrange the
            right-hand side as follows
            \[
                U - \frac UL\sum_{j=k_{\max} + 1}^{b}t_j = U - \frac UL \left(L - \sum_{j=1}^{k_{\max}}t_j\right) = \frac UL \sum_{j=1}^{k_{\max}}t_j\,.
            \] 
            By Lemma \ref{lem-prefix-sums} it holds that
            \[
                 \frac UL \sum_{j=1}^{k_{\max}}t_j = \frac UL \left(U - (m - 1)t_{k_{\max}}\right) = \frac ULs_1
            \] 
            due to the choice of $s_1$ in the configuration $\mathcal{S}_{k_{\max}}$.
            Thus the makespan would be at least $U/L=\overline{\rho}(m,b)$.

    The makespan was at least $U/L$ in both cases, hence the
    robustness factor is at least $U/L=\overline{\rho}(m,b)$ and the
    theorem follows.
\end{proof}
    
\section{Pebbles}
\label{sec-pebbles}

Recall that an instance of our problem is called $q$-pebbles if the processing times satisfy
    \[
        p_j \le q \cdot \frac{P}{m} = q \cdot \frac{\sum_{\ell=1}^n p_\ell}{m} \,.
    \] 

This definition might seem a bit unnatural at the first glance, but there is a very intuitive formulation.
The expression $\frac{P}{m}$ represents the \textit{average load of a machine}.
The definition of pebbles says that the processing times are relatively small compared to the average load of all machines.

Without loss of generality we assume in this section that the sum of processing times is $P=m$.
This transforms the condition for $q$-pebbles from the definition into
\[
    p_j \le q\,,
\]
which is easy to work with.

We use similar ideas as in the optimal algorithm for sand.
Recall the condition of Theorem~\ref{thm-greedy}
\[
    a_k \le \frac{\rho P - \sum_{j = 1}^{k - 1}a_j}{m}\,.
\]
As we have already noticed in Section~\ref{sec-sand-upper}, the optimal bag sizes for sand not only satisfy the above inequality, they actually have equality there.
The bag sizes for sand are given by the recurrence
\[
    a_k = \frac{\overline\rho(m, b) P - \sum_{j=1}^{k-1}a_j}{m}\,.
\]
When we in addition assume $P = m$, as in the case of pebbles, we get
\begin{equation}
    a_k = \overline\rho(m, b) - \frac 1 m \sum_{j=1}^{k-1}a_j. \label{eq-sand}
\end{equation}
Let $a_1, \ldots, a_b$ denote values given by the recurrence (\ref{eq-sand}) for the rest of this section.
Remember that the sum of $a_1, \ldots, a_b$ equals $P$.
Let us denote the bag sizes we will be choosing for pebbles $d_1, \ldots, d_b$.
We again want to use Theorem~\ref{thm-greedy}.
In other words, for the desired robustness factor $\rho$, we want the
bag sizes to satisfy 
\begin{equation}
    d_k \le \rho - \frac 1 m \sum_{j=1}^{k-1}d_j\,. \label{ineq-pebbles}
\end{equation}

Consider the following algorithm.
Place as many pebbles as you can into the first bag while it satisfies the inequality (\ref{ineq-pebbles}).
Then do the same for the second bag
and so on until the last bag (or until we run out of jobs).
See \nameref{alg-pebbles} for pseudocode.

\begin{algorithm}
    \caption{\normalfont \scshape Pebbles}
    \begin{algorithmic}
        \State {\bf Input}: processing times $p_1 \ge \cdots \ge p_m$;
        number of machines $m$;
        number of bags $b$;
        desired robustness factor $\rho$
        
        \State $B \gets $ empty mapping
        \For{$k \gets 1$ to $b$}
            $d_k \gets 0$
            \Comment{$d_k$ represents the size of the $k$th bag}
        \EndFor
        \State $k \gets $ 1
        \Comment{$k$ represents index of currently considered bag}
        \For{$j \gets 1$ to $n$}
            \While{$k \le b$ {\bf and} $d_{k} + p_j > \rho - \frac{1}{m}\sum_{\ell=1}^{k - 1}d_\ell$}
                $k \gets k + 1$
            \EndWhile
            \If{$k > b$}
                {\bf break}
            \EndIf
            \State $B[j] \gets k$
            \State $d_k \gets d_k + p_j$
        \EndFor
        \State \Return $B$
    \end{algorithmic}
    \label{alg-pebbles}
\end{algorithm}

\begin{theorem}
    \label{thm-pebbles}
    There exists a $(\overline{\rho}(m, b) + q)$-robust algorithm for $q$-pebbles, for $\overline{\rho}$ defined by {\rm (\ref{eq-rho})}.
\end{theorem}

\begin{proof}
We show that
    Algorithm \nameref{alg-pebbles} puts every job in some bag for $\rho = \overline{\rho}(m, b) + q$.

    Suppose for a contradiction that the algorithm does not use all the jobs.
    Then the bag sizes $d_k$ at the end of the algorithm must satisfy
    \[
        d_k + q > \rho - \frac 1 m \sum_{j=1}^{k-1}d_j\,.
    \]
Indeed, if for some $k$ this inequality is not satisfied, adding one more job of size at most $p$ to bag $k$ would not violate the inequality (\ref{ineq-pebbles}) and the algorithm would have done so.

Plugging in the expression for $\rho$ gives us
    \begin{equation}
      d_k > \overline{\rho}(m, b) - \frac 1 m \sum_{j=1}^{k-1}d_j.
      \label{ineq-pebbles-better}
    \end{equation}
    We are going to show
    \begin{equation}
        \sum_{j=1}^kd_j \ge \sum_{j=1}^ka_j\,,
      \label{ineq-pebbles-claim}
    \end{equation}
    for all $k \in \{0, \ldots, b\}$.
    We prove this claim by induction.
    The case $k = 0$ is trivial since the summations are empty and both sides are equal to 0.
    Let us now prove the induction step for $k$ using the equation (\ref{eq-sand}) and the inequality (\ref{ineq-pebbles-better}).
    \[
        d_k - a_k \ge \left(\overline{\rho}(m, b) - \frac 1 m \sum_{j=1}^{k-1}d_j\right) - \left(\overline{\rho}(m, b) - \frac 1 m \sum_{j=1}^{k-1}a_j\right) 
        = -\frac{1}{m}\left(\sum_{j=1}^{k-1}d_j - \sum_{j=1}^{k-1}a_j\right)
    \]
    We can now easily finish the induction step. We simplify
    \begin{align*}
        \sum_{j=1}^kd_j &- \sum_{j=1}^ka_j = \left(\sum_{j=1}^{k-1}d_j - \sum_{j=1}^{k-1}a_j\right) + \left(d_k - a_k\right) \\
        &\ge \left(\sum_{j=1}^{k-1}d_j - \sum_{j=1}^{k-1}a_j\right) - \frac{1}{m}\left(\sum_{j=1}^{k-1}d_j - \sum_{j=1}^{k-1}a_j\right) 
        = \frac{m-1}{m}\left(\sum_{j=1}^{k-1}d_j - \sum_{j=1}^{k-1}a_j\right),
    \end{align*}
    which is non-negative by the induction hypothesis for $k-1$ and
    thus the claim~(\ref{ineq-pebbles-claim}) holds.
    
    Using the claim~(\ref{ineq-pebbles-claim}) for $k = b$ gives us
    \[
        \sum_{j=1}^bd_j \ge \sum_{j=1}^ba_j = P\,,
    \]
    which is a contradiction with the assumption that we did not use all jobs.
\end{proof}

It is interesting to take a look at the case $b = m$.
Theorem \ref{thm-pebbles} implies that there exists an algorithm with robustness factor at most
\[
    \frac{e}{e - 1} + q \approx 1.58 + q\,.
\]
The best know result for rocks gives robustness factor $2 - 1/m$. 
This gets arbitrarily close to 2 for large $m$.
Hence we have obtained a stronger result for
\[
    q < 2 - \frac{e}{e - 1} \approx 0.42\,.
\]

%% file: arxiv-bricks.tex
\section{Bricks}
\label{sec-bricks}

In this section, we study the case of jobs with equal processing
times.  An important parameter is the ratio of the number of jobs and
the number of machines, which we denote $\lambda=n/m$. We can scale
the instance so that $p_j=1$ for all $j$, which we assume from now
on. Note that now $P=n$ and the average load is $P/m=n/m=\lambda$.

Thus the instance satisfies the definition of $p$-pebbles for
$p=1/\lambda$. Theorem~\ref{thm-pebbles} immediately implies our first
improved bound for bricks:

\begin{theorem}
    There exists an algorithm with robustness factor at most $\overline{\rho}(b, m) + m/n$ solving the problem for $n$ bricks, $m$ machines and $b$ bags. \qed
    \label{thm-bricks}
\end{theorem}

In the rest of this section we focus on our main result, the $1.6$-robust algorithm for bricks in case $b=m$. This will have the following ingredients:
\begin{itemize}
    \item For $\lambda\geq 60$ we have $e/(e-1)+1/60<1.6$, so by Theorem~\ref{thm-bricks} we can use Algorithm \nameref{alg-pebbles}.
    \item For $\lambda<60$ we design a new algorithm \nameref{alg-bag-size}. We split the analysis into two cases.
    \begin{itemize}
        \item For $m\geq 144$, we modify its solution into a certain fractional solution, which is easier to analyze, and bound the difference between the two solutions. 
        \item For $m<144$, we have a finite number of instances, which we verify using a computer. 
    \end{itemize}
    We stress that the analysis of instances for $m<144$ shows that
    Algorithm \nameref{alg-bag-size} works here without any changes,
    too, i.e., it does not lead to an algorithm with exploding number
    of cases tailored to specific inputs.
\end{itemize}

\subsection{First stage algorithm \nameref{alg-bag-size}}
\label{sec-bricks-first}

Our assumptions on the optimal solution explained at the beginning of
Section~\ref{sec-prelim} imply that we can also restrict ourselves to
instances with $\sum_{i=1}^m s_i = n$ and furthermore the values of
speeds $s_i$ are integral, as in the optimal solution the machine
loads are necessarily integral. (Recall that this is due to the fact
that we can modify the speeds independently of the first-stage
algorithm.)

The key ingredient of the improved algorithm is to observe that the
integrality of speeds allows us to use the pigeonhole principle to create
larger bags. Furthermore, with appropriate accounting we can use the
pigeonhole principle iteratively.

Let us demonstrate this on an example.  Let $n = 13$, $m = 10$ and
$\rho = 1.6$.  The total speed of $10$ machines is $13$, so one
machine has speed at least $2$.  This means that one of the machines
will have capacity $2\rho=3.2$ and we can create and assign a bag of
size $\lfloor 2\rho\rfloor=3$.  Without integrality of the speeds,
only a machine with speed $1.3$ would be guaranteed, so the capacity
would be just above $2$.

To continue iteratively, we cannot reason about the
capacity as in Algorithm \nameref{alg-greedy}. Instead, for each bag
we reserve some integral amount of speed on one of the machines.
For this accounting, we represent the remaining unreserved total speed
by \textit{coins}.

In the example above, we pay $2$ coins for a bag of size $3$. This
seems like an overpayment compared to Algorithm \nameref{alg-greedy},
as the $2$ coins correspond to capacity $3.2$, so we waste a capacity
of $0.2$. However, after this
step, we are left with $11$ coins among the $10$ machines, and using
the integrality and the pigeonhole principle once more, we are
guaranteed to have one machine with $2$ coins (these coins may be on a
different machine
or they may be the ones remaining on the same machine). Thus we can create
another bag of size $3$. Now there are only $9$ coins remaining and we
can only create a bag of size 1 at cost 1.  See Figure~\ref{fig-13_10}
for an illustration.  Overall, the effect of integrality is more
significant than the overpayment due to rounding, and thus we are able
to obtain an improved algorithm.

\begin{figure}[ht]
    \centering
    \includegraphics{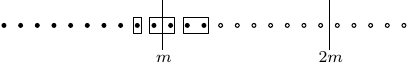}
    \caption{Graphical representation of the first three chosen bags for $n = 13$, $m = 10$. The dots represent coins and the boxes represent chosen bags. The number of coins inside a box represent the cost of the bag. Vertical lines emphasize the multiples of $m$, which determine the bag costs.}
    \label{fig-13_10}
\end{figure}

Formally, we start Algorithm \nameref{alg-bag-size} with $c=n$
coins. In each round we pay $z=\lceil c/m\rceil$, create a bag of size
$\lfloor z\cdot\rho\rfloor$ and continue with remaining coins on $m$
machines.  The \textit{cost} of a bag is the number of coins we pay
for it, i.e., $z$ in the algorithm.

\begin{algorithm}[ht]
    \caption{\normalfont \scshape Bricks}
    \begin{algorithmic}
        \State {\bf Input}: number of bricks $n$;
       number of machines $m$; 
       number of bags $b$; 
       desired robustness factor $\rho$
       \State $c \gets n$ \Comment{The initial number of coins is n}
        \For{$j \gets 1$ to $b$}
            \State $z \gets \lceil c/m\rceil$ 
            \Comment{max guaranteed coins on a machine}
            \State $a_j \gets \lfloor z \cdot \rho \rfloor$ 
            \Comment{max integer such that $\cost(a_j) = z$}
            \State $c \gets c - z$
        \EndFor
        \State \Return $a_1, a_2, \ldots, a_b$
    \end{algorithmic}
    \label{alg-bag-size}
\end{algorithm}

If Algorithm \nameref{alg-bag-size} produces bags of total size at least $n$, we say it is {\em successful}.
If the total sum of bag sizes exceeds $n$, we decrease the sizes of some bags to make the sum equal to $n$.
E.g., we can remove some of the last small bags and then decrease size of the last non-empty bag as needed.

In Section~\ref{sec-bricks-second} we show that this algorithm is
sound, namely, we give a modification of the second stage algorithm
algorithm \nameref{alg-greedy} for which we show that a machine
with unused speed $z$ always exists and thus we can assign all bags.

For a general instance, there is always a machine of speed at least $\lceil n/m\rceil=\lceil \lambda\rceil$, and thus the cost of the first bag is chosen as $\lceil\lambda\rceil$.
The cost will then decrease by 1 every time the number of coins decreases below a multiple of $m$.
Figure~\ref{fig-lambda-bag} illustrates this.

\begin{figure}[h]
    \centering
    \includegraphics{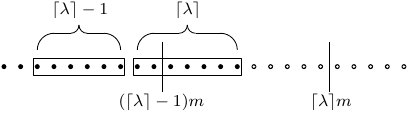}
    \caption{Graphical representation of the first chosen bag of size $\lceil\lambda\rceil$.}
    \label{fig-lambda-bag}
\end{figure}

Note that the costs of the bags chosen by \nameref{alg-bag-size} do
not depend on $\rho$. The sizes of the bags, however, do depend on
$\rho$.  See Figure~\ref{fig_45_9} below for an example execution of
\nameref{alg-bag-size} for $n = 45$ and $m = 9$. This execution shows that
\nameref{alg-bag-size} fails for $\rho < 1.6$ but succeeds for $\rho = 1.6$.

\begin{figure}
    \centering
    \begin{tabular}{c|c|c}
         remaining coins $c$ & $a_j$ for $\rho = 1.6$ & $a_j$ for $\rho < 1.6$ \\
         \hline
         45 & 8 & $<8$ \\
         40 & 8 & $<8$ \\
         35 & 6 & $\le6$ \\
         31 & 6 & $\le6$ \\
         27 & 4 & $\le4$ \\
         24 & 4 & $\le4$ \\
         21 & 4 & $\le4$ \\
         18 & 3 & $\le3$ \\
         16 & 3 & $\le3$ \\
         \hline
          & 46 & $\le44$ \\
    \end{tabular}
\medskip

\centering
        \includegraphics{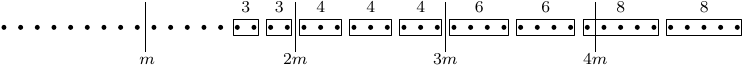}
        \caption{
 Tabular and graphical
 representation of the execution of \nameref{alg-bag-size} for $n = 45$, $m = 9$ and $\rho = 1.6$.
            The numbers above bags represent their sizes.
            The sum of bag sizes is actually $46 > n = 45$, to solve this, we can for example replace one bag of size 3 with a bag of size 2.
        }
        \label{fig_45_9}
\end{figure}

\subsection{Fractional solutions}

In general, the cost of the first bag chosen will be $\lceil\lambda\rceil$.
The cost will then decrease by $1$ every time the number of coins
decreases below a multiple of $m$. Roughly speaking, we use
approximately $m$ coins for bags of each size.

We need to show that the created $b$ bags have total size at least
$n$. If we would use exactly $m/z$ bags for each cost $z$, the total
size of bags is easy to compute. However, the integral number of bags
of each cost causes rounding issues when the bag cost decreases and
these complicate the calculations.

To structure our analysis, we first modify the solution obtained by
\nameref{alg-bag-size} into a solution that uses possibly non-integral
number of bags of each size.  In such a solution, we can use fractions
of bags (such as $\frac{4}{5}$ of a bag of size 8 as in the
Figure~\ref{fig_45_9_fractional}).  We arrange the modification so
that the total size of bags of cost $z$ is exactly $m$, except for the
smallest and largest bag costs.  In the main part of our proof, we
bound the rounding error, i.e., the difference between the sizes of
the integral and fractional solution. To complete the proof, we
calculate the total size of bags in the modified fractional solution,
which is easy, and show that it is well above $n$.

\begin{figure}[ht]
    \centering
    \includegraphics{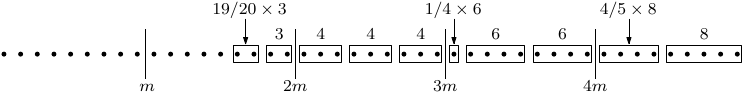}
    \caption{
        Fractional solution for $n = 45$, $m = 9$ and $\rho = 1.6$ produced by \nameref{alg-fractional}.
        Notice that we always use only one bag size (cost) between consecutive multiples of $m$.
        Compare this to Figure~\ref{fig_45_9} where bag cost $5$
        ``overflows'' the line at $4m$ coins.
    }
    \label{fig_45_9_fractional}
\end{figure}

For the fractional solutions, it is better to use an
alternative representation of the bags by a function $F$ that for each
$z$ gives the number $F(z)$ of bags of cost $z$.  The size of $F$ is
then defined as the total size of bags.  Recall that a bag of cost $z$
has size $\lfloor z \cdot \rho \rfloor$. Formally:

\begin{definition}
    A fractional solution is a mapping $F: \mathbb{N} \to \mathbb{R}_0^+$ satisfying $\sum_{z=1}^\infty F(z) = b$.
   The size of fractional solution $F$ for robustness factor $\rho$ is defined as
    \[
        \size(F, \rho) = \sum_{z=1}^\infty F(z) \cdot \lfloor z \cdot \rho \rfloor\,.
    \]
    We will sometimes use only $\size(F)$ if $\rho$ is clear from the context. 
\end{definition}

We start by reformulating \nameref{alg-bag-size} so that it produces
the solution directly in the alternative representation, see
Algorithm \nameref{alg-integral-fractional} below. It is easy to see
that \nameref{alg-bag-size} and \nameref{alg-integral-fractional} are
equivalent.

\begin{algorithm}[ht]
    \caption{\normalfont \scshape BricksAlt}
    \begin{algorithmic}
        \State {\bf Input}: number of bricks $n$;
        number of machines $m$;
       number of bags $b$ 
        
        \State {\bf Output}: Fractional solution $I$
      
        \State $r \gets b$
        \Comment{$r$ is the remaining number of bags (integral)}
        \State $c \gets n$
        \Comment{$c$ is the remaining number of coins (integral)}
        \State $I[z] \gets 0$ for $z \in \mathbb{N}$
        \While{$r > 0$ {\bf and} $c > 0$}
            \State $z \gets \lceil \frac c m \rceil$
            \Comment{$z$ is the bag cost}
            \State $x \gets \min\left(r, \lceil\frac{c - m(z - 1)}{z}\rceil\right)$
            \Comment{$x$ is the (integral) number of bags of cost $z$}
            \State $r \gets r - x$
            \State $c \gets c - x \cdot z$
            \State $I[z] \gets x$
        \EndWhile
        \State \Return $I$
    \end{algorithmic}
    \label{alg-integral-fractional}
\end{algorithm}

\begin{observation}
    \nameref{alg-bag-size} and \nameref{alg-integral-fractional} use each bag cost the same number of times.
\end{observation}
\begin{proof}
    One step of \nameref{alg-integral-fractional} corresponds to several steps of \nameref{alg-bag-size}.
    \nameref{alg-bag-size} chooses the bags one by one, and it may choose the same bag cost in several consecutive iterations.
    \nameref{alg-integral-fractional} in each step calculates how many bags of given cost would \nameref{alg-bag-size} use.
    The key observation is that the expression
    $
        \lceil(c - m(z - 1))/z\rceil
    $
    calculates how many bags of cost $z$ are needed to have at most $m(z - 1)$ coins remaining.
    In other words, it calculates how many bags of cost $z$ \nameref{alg-bag-size} uses before it starts using bags of cost $z - 1$ (or runs out of bags).
    Hence both \nameref{alg-bag-size} and
    \nameref{alg-integral-fractional} use the same number of bags of
    cost $z$ for each $z$.
\end{proof}

Algorithm \nameref{alg-fractional} (see below) is obtained from \nameref{alg-integral-fractional} by removing the rounding in the calculation of the number of bags $x$.
\begin{algorithm}[ht]
    \caption{\normalfont \scshape BricksFract}
    \begin{algorithmic}
        \State {\bf Input}: number of bricks $n$;
        number of machines $m$;
       number of bags $b$
        
        \State {\bf Output}: Fractional solution $F$

        \State $r \gets b$
        \Comment{$r$ is the remaining number of bags (fractional)}
        \State $c \gets n$ 
        \Comment{$c$ is the remaining number of coins (fractional)}
        \State $F[z] \gets 0$ for $z \in \mathbb{N}$
        \While{$r > 0$ {\bf and} $c > 0$}
            \State $z \gets \lceil \frac c m \rceil$
            \State $x \gets \min\left(r, \frac{c - m(z - 1)}{z}\right)$
            \Comment{$x$ is the fractional amount of bags of cost $z$}
            \State $r \gets r - x$
            \State $c \gets c - x \cdot z$
            \State $F[z] \gets x$
        \EndWhile
        \State \Return $F$
    \end{algorithmic}
    \label{alg-fractional}
\end{algorithm}

The following observation says that the algorithm follows our initial intuition, namely that for bags of each cost we use exactly $m$ coins, except for the first and last bag cost used.
\begin{definition}
    Let $F$ be a fractional solution, then let $z_{\min}$ and $z_{\max}$ denote the smallest and largest integers such that $F(z_{\min}) > 0$ and $F(z_{\max}) > 0$.
\end{definition}
\begin{observation}
    \label{obs-Fz}
    Let $F$ be a result \nameref{alg-fractional} with input $n$ and $m$.
    Then $F(z) = m/z$
     for every $z$ such that $z_{\min} < z < z_{\max}$.
\end{observation}
\begin{proof}
    Observe that in every step of the algorithm, except the last one, it holds that   $x = (c - m(z - 1))/z$
    and thus
    $c - x \cdot z = m(z - 1)$.
    It follows that in all the steps except for the first and last ones 
        $x = (mz - m(z - 1))/z = m/z$.
\end{proof}

Next we observe that the result of \nameref{alg-fractional} scales, i.e., essentially it depends only on $\lambda$. 
Note also that \nameref{alg-fractional} is well defined even for non-integral $m$ and $n$.
\begin{observation}
    Let $\alpha \in \mathbb{R}^+$.
    Suppose \nameref{alg-fractional} produces solution F with $n$ and $m$ as an input and solution $\bar{F}$ with inputs $\alpha n$ and $\alpha m$.
    Then $
        \bar{F}(z) = \alpha F(z)
    $
    for all $z$.
    It follows that
    $
        \size(\bar{F}, \rho) = \alpha \cdot \size(F, \rho)
    $.
    \label{obs-invariant}
\end{observation}
\begin{proof}
    We go through the execution of \nameref{alg-fractional} step by step.
    Suppose that we multiply both $m$ and $n$ by $\alpha$.
    Then in every iteration of the loop 
 $r$ is multiplied by $\alpha$,
 $c$ is multiplied by $\alpha$,
 $z$ stays the same, and
 $x$ is multiplied by $\alpha$.
\end{proof}

Now we are ready to bound the difference between the solutions
produced by 
\nameref{alg-fractional} and \nameref{alg-integral-fractional}, i.e.,
the rounding error.
\begin{lemma}
    \label{lem-difference}
    Let $F$ be the fractional solution produced by \nameref{alg-fractional} and  $I$ the solution produced by \nameref{alg-integral-fractional} on the same input. Then
    \begin{itemize}
        \item 
            for $\lambda \le 5$ it holds that $\size(I, 1.6) \ge \size(F, 1.6)$ and
        \item 
            for $\lambda \le 60$ it holds that
            $\size(I, 1.6) \ge \size(F, 1.6) - 12$.
    \end{itemize}
\end{lemma}
\begin{proof}
    We give an algorithm which transforms  $F$ into a solution $F'$ that is almost integral and very close to $I$. Set $\bar{z}$ to be $z_{\min}$ of the solution $F$ and note that $z_{\max} = \lceil\lambda\rceil$.
    
   We go through the bag costs, denoted by $z$, from $\lceil\lambda\rceil$ down to $\bar{z}+1$. For each $z$, 
    if $F$ uses non-integral amount of bags of cost $z$, round it up. This makes the number of bags of cost $z$ equal to their number in the solution $I$. Then decrease the number of bags of cost $z-1$ so that the total cost of all bags remains the same. Finally, increase the number of bags of cost 1 so that the total number of bags stays equal to $b$. See Figure~\ref{fig_45_9_transformed} for an illustration.

\begin{figure}[h]
    \centering
    \includegraphics{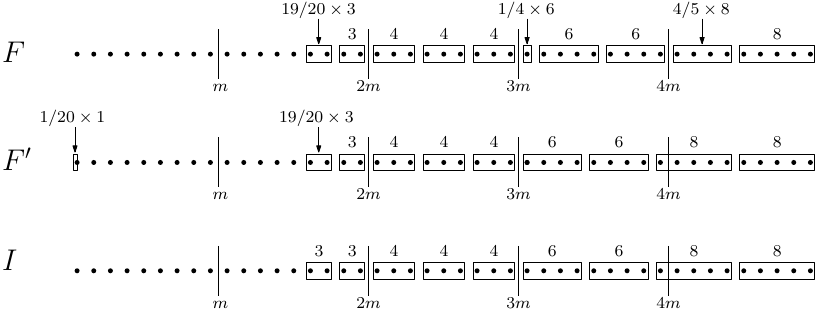}
    \caption{Graphical representation of $F$, $F'$ and $I$.  In the
      first step of the transformation from $F$ to $F'$, $\frac{9}{5}$
      is rounded up to 2 and the number of bags of cost 5 (and
      size 8) increases by $\frac 1 5$.  In order to keep the total
      cost the same, number of bags of cost 4 (and size 6) is decreased
      by $\frac 1 4$.  As a result, total number of bags decreased by $\frac 1
      4 - \frac 1 5 = \frac{1}{20}$, hence we add $\frac{1}{20}$ of a
      bag of cost 1 (and size 1).  This is actually the only step in
      which something happens since number of used bags of cost $4$
      and $3$ is already integral. We do not process the bags of cost
      2, as $\bar{z}=2$.  Solution $F'$ is almost identical
      to the solution $I$, but has $\frac{1}{20}$ of bag of cost $1$
      instead of $\frac{1}{20}$ of bag of cost 2.  }
    \label{fig_45_9_transformed}
\end{figure}

For $z=\bar{z}+1$, the previous procedure could lead to a negative
value of $F'(\bar{z})$. In this special case we proceed slightly
differently and instead of rounding $G(z)$ up we only increase it so that
$F'(\bar{z})=0$.

We now describe one step of the process formally and analyze it.  Let
$G$ denote the current fractional solution and let $H$ denote the
result of one transformation step.  Let $z$ be the current cost of bags.
    
We set $H(z) = \lceil G(z) \rceil$, note that $H(z)-G(z)<1$.  We want
the sum of costs of bags of costs $z-1$ and $z$ to remain the same,
hence we want
    \begin{equation}      
      H(z)\cdot z +
        H(z-1)\cdot (z-1) = 
        G(z)\cdot z +
        G(z-1)\cdot(z-1)
\label{eq-brick-cost}
    \end{equation}
    to hold. Rearranging (\ref{eq-brick-cost}) to an equivalent equation
    leads to (\ref{eq-brick-H}), so we set
    \begin{equation}
         H(z-1) = G(z-1) + (G(z) - H(z)) \cdot \frac{z}{z-1}\,.
\label{eq-brick-H}
    \end{equation}
We claim that $H(z-1)>0$ for $z-1>\bar{z}$. Indeed, as $m>144$ and $z\leq \lceil\lambda\rceil\leq 60$ in the considered case, we have $G(z-1)=F(z-1)=m/(z-1)>2$ (using also $\bar{z}=z_{\min}<z-1<z_{\max}$). As $H(z)-G(z)<1$, we get $H(z-1)>G(z-1)-1>0$.

 Now we describe the modification in the special case when $H(z-1)$
 would become negative. We have shown above that this can happen only
 for $z=\bar{z}+1$, i.e., in the last step. Then we set $H(z-1)=0$ and set
\[
H(z)=G(z)+(G(z-1)-H(z-1))\frac{z-1}{z}.
\]
This equation is equivalent to (\ref{eq-brick-cost}), which is in turn
equivalent to (\ref{eq-brick-H}), which thus also holds. Furthermore,
the fact that the previous procedure would lead to negative $H(z-1)$
implies that now we have $G(z)\leq H(z)\leq \lceil G(z) \rceil$ and
thus $H(z)-G(z)<1$ holds again.

In both cases, the total number of bags has decreased by
    \[
        (G(z) - H(z)) + (G(z-1) - H(z-1)) = \frac{1}{z - 1}(H(z) - G(z))\,.
    \]  
    Thus we set
    \[
        H(1) = G(1) + \frac{1}{z - 1}(H(z) - G(z))\,.
        \]
Note that in the transformation step, both the total number of bags
and their total cost remain constant.

    Recall that the size of a bag of cost $z$ is $\lfloor z\rho\rfloor$.
    It follows that
    \begin{align*}
        \size&(H) - \size(G) \\
       & =(H(z) - G(z)) \cdot \lfloor z\rho\rfloor
        +(H(z-1) - G(z-1)) \cdot \lfloor (z-1)\rho\rfloor
        +(H(1) - G(1)) \cdot \lfloor \rho\rfloor\\
       & =(H(z) - G(z)) \cdot \left(\lfloor z\rho\rfloor - \frac{z}{z-1}\lfloor (z-1)\rho\rfloor + \frac{1}{z - 1}\lfloor\rho\rfloor\right)
    \end{align*}
    Note that the second factor in the expression above does not depend on the solution. We call it the transformation factor and for $z$ we denote it by 
    \[
        f(z) = \left(\lfloor z\rho\rfloor - \frac{z}{z-1}\lfloor (z-1)\rho\rfloor + \frac{1}{z - 1}\lfloor\rho\rfloor\right)\,.
    \]  
    If $f(z) \geq 0$, the size of the solution could have only
    increased, as $H(z)\geq G(z)$, i.e., we have $\size(H)\geq
    \size(G)$.  If $f(z) < 0$, the size of the solution might have
    decreased---those are the important (``bad'') cases we need to
    bound.  We have $H(z) - G(z) < 1$, hence
    $\size(H)\ge\size(G) + f(z)$ in case of negative $f(z)$.

    Now we sum these bounds over all steps for $z$ from $\lceil\lambda\rceil$ to $2$ and get
    \[
        \size(F') - \size(F) \ge \sum_{z=2}^{\lceil\lambda\rceil} \min (0, f(z))
    \]  
    We give a list of values of $f(z)$ for $z$ from 2 to $60$ and $\rho=1.6$
    in Appendix \ref{apx-f}.
    For $z\leq 5$ the values $f(z)$ are non-negative, thus for $\lambda\leq 5$ we get 
    $\size(F') - \size(F) \ge 0$. It can be verified that the sum of all negative values of $f(z)$ for $z\leq 60$ is larger than $-12$
    and thus for $\lambda\leq 60$ we get
    $\size(F') - \size(F) > -12$.
    
Examining the algorithms \nameref{alg-integral-fractional} and
\nameref{alg-fractional} that generate the solutions $I$ and $F$,
respectively, and the transformation process above shows that the
solution $F$ is step by step transformed towards $I$.  In particular,
$I(z)=F'(z)$ for all values $z\geq \bar{z}$ (if the special case does
not apply) or ($z\geq \bar{z}+1$ if the special case applies). For the
small values of $z$, the only possible difference is that solution
$F'$ might have some amount of bags of size 1 instead of some larger
bags in solution $I$. (Note that the total number of bags does not
change during the transformation.)  This implies $\size(I)
\ge\size(F')$ and the lemma follows.
\end{proof}

To complete the proof we need to show that $\size(F)$ is sufficiently
large so that $\size(I)\geq n$. Actually, as the previous
transformation possibly gives $I$ with a slightly smaller size than
$F$, we need to compensate for this difference which is at most
$12$. Precisely, we need to prove that $\size(F, 1.6)\geq n$ for
$\lambda \leq 5$ and $\size(F, 1.6)\geq n+12$ for
$5<\lambda \leq 60$ and $m\geq 144$.

Since the fractional solution $F$ scales when $m$ and $n$ are scaled,
see Observation~\ref{obs-invariant}, it is convenient to normalize by $m$ and
consider $(\size(F, 1.6)-n)/m$ in the following lemma. Let us call
this crucial quantity \textit{normalized brick surplus}, as it
measures how many bricks we are able to put in the bags in the fractional
solution in addition to $n$ bricks, normalized by $m$.

\begin{lemma}
    \label{lem-surplus}
    Let $F$ be a fractional solution produced by \nameref{alg-fractional}. Then
    \begin{itemize}
        \item 
            For $\lambda \le 4$ it holds that 
            $\size(F, 1.6) \ge n$.
       \item 
            For $4 \le \lambda \le 60$ it holds that
            $\size(F, 1.6) \ge n + \frac{1}{12}m$.
    \end{itemize}
\end{lemma}

\begin{proof}
By Observation~\ref{obs-invariant}, the normalized brick surplus
$\size(F, 1.6) - n)/m$ is uniquely determined by $\lambda$, i.e.,
multiplying both $n$ and $m$ by the same constant does not change it.
    
This means that the normalized brick surplus is a function of
$\lambda$.  Furthermore, we claim that the function is piece-wise
linear.  Suppose we slowly increase $\lambda$, for example fix $m$ and
increase $n$ by $\delta$. Then $F(z)$ remains constant for all $z$
except $z_{\min}$ and $z_{\max}$ by Observation \ref{obs-Fz}.  The
number of the largest bags $F(z_{\max})$ increases by
$\delta/\lceil\lambda\rceil$ and $F(z_{\min})$ decreases by the same
amount; this amount is proportional to $\delta$. The function
$\size(F)$ is linear in the values of $F(z)$. So the normalized brick
surplus is piece-wise linear with possible breakpoints between the segments
at the values of $\lambda$ when one of the values of $z_{\min}$ or
$z_{\max}$ changes. 

The value of $z_{\max}$ changes exactly when $\lambda$ is an integer.
The breakpoints where $z_{\min}$ increases can be calculated in the
following way: Execute \nameref{alg-fractional} for all integer values
of $\lambda \le 60$.  Let us denote one of such solutions $F$.  Take a
look at $F(z_{\min})$, if we now slowly increase $\lambda$,
$F(z_{\min})$ will decrease linearly as described above.  Calculate at
which point it reaches 0; if it happens before $\lambda$ increases
above another integer, we found a point where $z_{\min}$ changes.  The first
case of changing $z_{\min}$ is at $\lambda = \frac{11}{3}$ when we
stop using bags of cost 1, see Figure~\ref{fig-11_3}.

\begin{figure}[htb]
    \centering
    \includegraphics{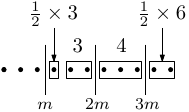}
   \caption{
        Example of solution produced by \nameref{alg-fractional} for $n = 11$, $m = 3$ and $\rho = 1.6$.
        Size of this solution is $11.5$ and normalized brick surplus is $\frac 1 6$.
        The solution does not use any bags of cost 1.
        However, if $\lambda$ were smaller,
    the solution would use bags of size $1$.
       }
   \label{fig-11_3}
\end{figure}

The computer-generated tables of values of the normalized brick
surplus function are given in Appendix~\ref{apx-brick-surplus}. The
plot of the values is given in Figures~\ref{fig-small-lambda}
and~\ref{fig-large-lambda} below.

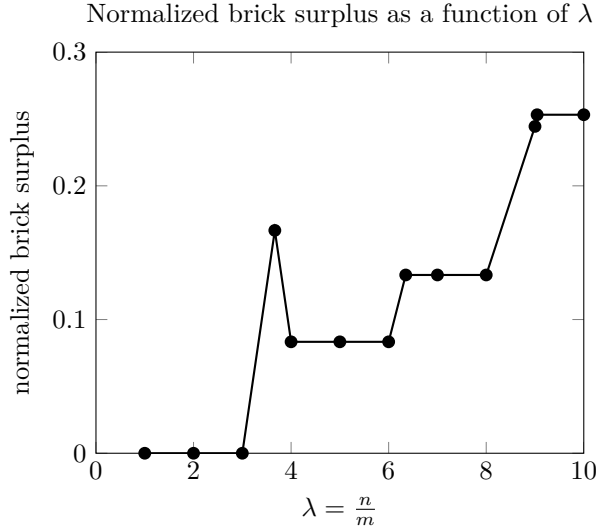
\begin{figure}[ht]
    \centering
    \begin{tikzpicture}
        \begin{axis}[
            title={Normalized brick surplus as a function of $\lambda$},
            xlabel={$\lambda = \frac n m$},
            ylabel={normalized brick surplus},
            xmin=0, xmax=10,
            ymin=0, ymax=0.3,
            legend pos=north west,
        ]
        
            \addplot[
                color=black,
                mark=*,
                style={line width=0.8pt}
            ]
            coordinates {
                (1.0, 0.0)
                (2.0, 0.0)
                (3.0, 0.0)
                (3.6666666666666665, 0.16666666666666666)
                (4.0, 0.08333333333333333)
                (5.0, 0.08333333333333333)
                (6.0, 0.08333333333333333)
                (6.35, 0.13333333333333333)
                (7.0, 0.13333333333333334)
                (8.0, 0.13333333333333333)
                (9.0, 0.24444444444444444)
                (9.043650793650794, 0.25317460317460316)
                (10.0, 0.25317460317460316)
            };
        \end{axis}
        \end{tikzpicture}
    \caption{Plot of normalized brick surplus for small $\lambda$.}
    \label{fig-small-lambda}
\end{figure}

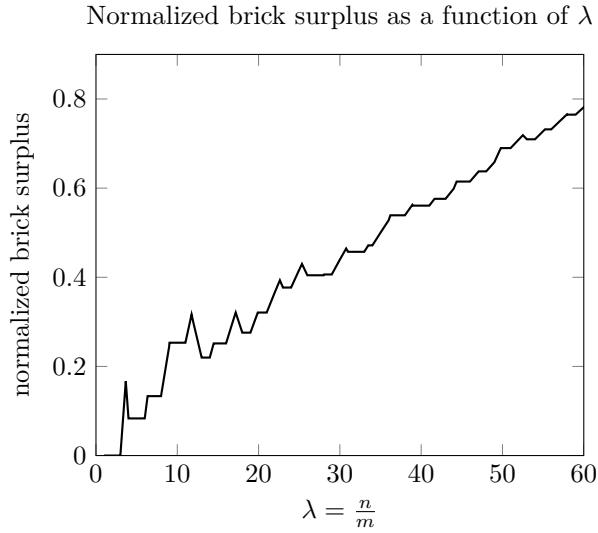
\begin{figure}[!ht]
    \centering
    \begin{tikzpicture}
        \begin{axis}[
            title={Normalized brick surplus as a function of $\lambda$},
            xlabel={$\lambda = \frac n m$},
            ylabel={normalized brick surplus},
            xmin=0, xmax=60,
            ymin=0, ymax=0.9,
            legend pos=north west,
        ]
        
            \addplot[
                color=black,
                style={line width=0.8pt}
            ]
            coordinates {
                (1.0, 0.0)
                (2.0, 0.0)
                (3.0, 0.0)
                (3.6666666666666665, 0.16666666666666666)
                (4.0, 0.08333333333333333)
                (5.0, 0.08333333333333333)
                (6.0, 0.08333333333333333)
                (6.35, 0.13333333333333333)
                (7.0, 0.13333333333333333)
                (8.0, 0.13333333333333333)
                (9.0, 0.24444444444444444)
                (9.043650793650794, 0.25317460317460316)
                (10.0, 0.25317460317460316)
                (11.0, 0.25317460317460316)
                (11.761471861471861, 0.31663059163059165)
                (12.0, 0.29675324675324677)
                (13.0, 0.21983016983016984)
                (14.0, 0.21983016983016984)
                (14.476565101565102, 0.2516011766011766)
                (15.0, 0.2516011766011766)
                (16.0, 0.2516011766011766)
                (17.0, 0.3104247060129413)
                (17.188054592466358, 0.3208721833721834)
                (18.0, 0.27576410517586986)
                (19.0, 0.27576410517586986)
                (19.902349714269217, 0.3208815908893308)
                (20.0, 0.3208815908893308)
                (21.0, 0.3208815908893308)
                (22.0, 0.3663361363438763)
                (22.622009530716962, 0.39338002898374413)
                (23.0, 0.3769456607540468)
                (24.0, 0.3769456607540468)
                (25.0, 0.4169456607540468)
                (25.338261981583425, 0.4299557369687939)
                (26.0, 0.40450427472200257)
                (27.0, 0.40450427472200257)
                (28.0, 0.40450427472200257)
                (28.05211923505469, 0.4063014897238884)
                (29.0, 0.4063014897238884)
                (30.0, 0.4396348230572218)
                (30.771596632665567, 0.4645250370141756)
                (31.0, 0.45715718645500036)
                (32.0, 0.45715718645500036)
                (33.0, 0.45715718645500036)
                (33.49002338401834, 0.47156963892612813)
                (34.0, 0.47156963892612813)
                (35.0, 0.5001410674975567)
                (36.0, 0.5279188452753345)
                (36.20625865854729, 0.5390679619535665)
                (37.0, 0.5390679619535665)
                (38.0, 0.5390679619535665)
                (38.922752217412075, 0.5627282752205428)
                (39.0, 0.5607475628464933)
                (40.0, 0.5607475628464933)
                (41.0, 0.5607475628464933)
                (41.64241983638748, 0.5760432732366716)
                (42.0, 0.5760432732366716)
                (43.0, 0.5760432732366716)
                (44.0, 0.5987705459639443)
                (44.3601394954952, 0.6147767457637309)
                (45.0, 0.6147767457637309)
                (46.0, 0.6147767457637309)
                (47.0, 0.6360533415084118)
                (47.07625668367867, 0.6376420224183841)
                (48.0, 0.6376420224183841)
                (49.0, 0.6580501856836902)
                (49.79513699334442, 0.6898556654174671)
                (50.0, 0.6898556654174671)
                (51.0, 0.6898556654174671)
                (52.0, 0.7090864346482363)
                (52.513872441650335, 0.718782141094469)
                (53.0, 0.7096099230123999)
                (54.0, 0.7096099230123999)
                (55.0, 0.7277917411942181)
                (55.23113693218527, 0.7319191864118122)
                (56.0, 0.7319191864118122)
                (57.0, 0.749463046060935)
                (57.94803843916181, 0.7658085363913111)
                (58.0, 0.7649126474113424)
                (59.0, 0.7649126474113424)
                (60.0, 0.781579314078009)
                (60.66751307319777, 0.7925221513435462)
            };
        \end{axis}
        \end{tikzpicture}
    \caption{Plot of normalized brick surplus for large $\lambda$.}
    \label{fig-large-lambda}
\end{figure}

The lemma now follows, since the normalized brick surplus is always non-negative and it is at least $1/12$ for $\lambda \ge 5$.
(Note that it is a constant function equal to $1/12$ for $\lambda\in [4,6]$.)
\end{proof}

\begin{theorem}
\label{thm-bricks-alg}
For $\rho = 1.6$ and $\lambda \le 60$, Algorithm \nameref{alg-bag-size} always succeeds, i.e., outputs bags of total size at least $n$. 
\end{theorem}
\begin{proof}
For $\lambda\leq 5$, the first claims in Lemmata~\ref{lem-surplus} and~\ref{lem-difference} together prove 
$\size(I, 1.6) \ge \size(F,1.6)\ge n$.

For  $5 \le \lambda \le 60$ and $m\geq 144$  the second claims in Lemmata~\ref{lem-surplus} and~\ref{lem-difference} together prove 
$\size(I, 1.6) \ge \size(F,1.6)-12\geq n + m/12 - 12 \ge n +  144/12 - 12 = n$.

For $\lambda\leq 60$ and $m\leq 144$ there are only a finitely many instances and we verify $\size(I,1.6)\geq n$ for them by computer, see Appendix~\ref{apx-small-m}
\end{proof}

We note that our choice of the bounds in the previous two lemmata is
somewhat arbitrary. The plots of the normalized brick surplus suggest
that we could bound it by an appropriate linear function instead of a
constant. Also, the bound on $\size(F)-\size(I)$ can be made smaller
for intermediate values of $\lambda$. These changes would decrease the
number of cases we need to check by a computer program, but would not
improve the robustness factor.

\subsection{Second stage}
\label{sec-bricks-second}

We need to show that if \nameref{alg-bag-size} succeeds, in the second stage we can indeed achieve makespan $\rho$. To do this, we cannot use
Algorithm \nameref{alg-greedy} and Theorem~\ref{thm-greedy}. Instead we modify it to Algorithm \nameref{alg-integral} below, which copies the coins accounting scheme from 
\nameref{alg-bag-size} and thus follows the intuition behind it.

\begin{algorithm}[ht]
    \caption{\normalfont \scshape IntegralAssignment}
    \begin{algorithmic}
        \State {\bf Input}: bag sizes $a_1 \ge \cdots \ge a_b$;
        machine speeds $s_1\ge \cdots \ge s_m$;
        desired robustness factor $\rho$       
        \For{$i \gets 1$ to $m$}
           \State 
            $c_i \gets s_i$
            \Comment{ Machine $i$ gets $s_i$ coins at the beginning.}
           \State 
            $M_i=\emptyset$
             \Comment{Initialize the assignment}
         \EndFor
        \For{$k \gets 1$ to $b$}
            \State $i \gets $ index of the machine with the largest $c_i$
            \State $M_i \gets M_i\cup \{k\}$
            \Comment{Assign bag $k$ to machine $i$}
            \State $c_i \gets c_i - \lceil a_k/\rho\rceil$
            \Comment{Machine $i$ pays for the bag $k$}
        \EndFor
        \State \Return $M_1, \ldots, M_m$
    \end{algorithmic}
    \label{alg-integral}
\end{algorithm}

\begin{theorem}
    \label{thm-success-brick}
    Suppose the first-stage algorithm \nameref{alg-bag-size} succeeds, i.e., outputs bags of total size of at least $n$. Then \nameref{alg-integral} in the second stage produces an assignment with makespan at most $\rho$.
\end{theorem}
\begin{proof}
    Imagine that \nameref{alg-bag-size} and \nameref{alg-integral} are running in parallel.
    \nameref{alg-bag-size} chooses the size of one bag and \nameref{alg-integral} assigns it to a machine. Note that the values of $c_i$ remain integral during the entire execution.
    
    We claim that during the execution the value $c$ in \nameref{alg-bag-size} is at most $\sum_{i=1}^mc_i$ for $c_i$'s in \nameref{alg-integral}. 
    At the beginning, the quantities are equal.
    Suppose that  \nameref{alg-bag-size} creates a bag of cost $z$ and thus decreases $c$ by $z$. Then the bag has size $a=\lfloor z\cdot \rho\rfloor\leq z\rho$. Thus \nameref{alg-integral} decreases $c_i$ by 
    $\lceil a/\rho\rceil\leq \lceil z \rho/\rho \rceil=z$. Thus the sum of $c_i$'s decreases by at most $z$ and the claim follows.

The claim implies that in each step before creating/assigning a bag of cost $z$, there exists a machine with $c_i\geq z$ in \nameref{alg-integral}.
Indeed, \nameref{alg-bag-size} chooses $z=\lceil c/m \rceil$, thus $m(z-1)<c\leq \sum_{i=1}^mc_i$ using the previous claim. Hence there exists a machine with $c_i>z-1$ and together with integrality of $c_i$ we get $c_i\geq z$.

It follows that $c_i$'s remain non-negative during the execution.
Thus \nameref{alg-integral} assigned to machine $i$ bags of the total size at most $s_i\cdot \rho$. It follows that the makespan is at most $\rho$.
\end{proof}

Theorems~\ref{thm-bricks}, \ref{thm-bricks-alg}, and \ref{thm-success-brick} immediately imply our main result.
\begin{theorem}\label{thm-success}
    There exists 1.6-robust algorithm for the case of bricks and $b=m$.
    \qed
\end{theorem}

%% file: arxiv-bricks-computer.tex
\clearpage

\section{Computer-assisted proofs for bricks}
\label{sec-bricks-computer}

\subsection{Values of the transformation factor}
\label{apx-f}

Let us evaluate
\[
    f(z) = \left(\lfloor z\rho\rfloor - \frac{z}{z-1}\lfloor (z-1)\rho\rfloor + \frac{1}{z - 1}\lfloor\rho\rfloor\right)
\]
for $\rho = 1.6$ and $z \in \{2, 3, \dots, 60\}$.
See Figure~\ref{fig-table} for both exact and approximate values of $f(z)$.
You might notice that the signs of $f(z)$ are periodic (with the
exception of first few values) with a period of 5. As a result,
approximately $40 \%$ of the values of $f(z)$ are negative and the
table shows that they are all above $-0.6$. Examining and summing
these $22$ values leads to the bound of $12$ on $\size(F)-\size(I)$.

\begin{figure}[h]
\centering
\begin{tabular}{|c|r|r|c|c|r|r|c|c|r|r|}
\cline{1-3}
\cline{5-7}
\cline{9-11}
$z$ & $f(z)$ & $\approx f(z)$ & 
\hspace{3ex}
&$z$ & $f(z)$ & $\approx f(z)$ & 
\hspace{3ex}
&$z$ & $f(z)$ & $\approx f(z)$ \\
\cline{1-3}
\cline{5-7}
\cline{9-11}
&&&&&&&&&&\\
&&&&21 &$\mathbf{-\frac{11}{20}}$&\bf{--0.55}&&41 &$\mathbf{-\frac{23}{40}}$&\bf{--0.57}\\[1ex]
2 &$2$&2.00&&22 &$\frac{10}{21}$&0.48&&42 &$\frac{18}{41}$&0.44\\[1ex]
3 &$0$&0.00&&23 &$\mathbf{-\frac{6}{11}}$&\bf{--0.55}&&43 &$\mathbf{-\frac{4}{7}}$&\bf{--0.57}\\[1ex]
4 &$1$&1.00&&24 &$\frac{11}{23}$&0.48&&44 &$\frac{19}{43}$&0.44\\[1ex]
5 &$\frac{3}{4}$&0.75&&25 &$\frac{11}{24}$&0.46&&45 &$\frac{19}{44}$&0.43\\[1ex]
6 &$\mathbf{-\frac{2}{5}}$&\bf{-0.40}&&26 &$\mathbf{-\frac{14}{25}}$&\bf{--0.56}&&46 &$\mathbf{-\frac{26}{45}}$&\bf{--0.58}\\[1ex]
7 &$\frac{2}{3}$&0.67&&27 &$\frac{6}{13}$&0.46&&47 &$\frac{10}{23}$&0.43\\[1ex]
8 &$\mathbf{-\frac{3}{7}}$&\bf{-0.43}&&28 &$\mathbf{-\frac{5}{9}}$&\bf{--0.56}&&48 &$\mathbf{-\frac{27}{47}}$&\bf{--0.57}\\[1ex]
9 &$\frac{5}{8}$&0.62&&29 &$\frac{13}{28}$&0.46&&49 &$\frac{7}{16}$&0.44\\[1ex]
10 &$\frac{5}{9}$&0.56&&30 &$\frac{13}{29}$&0.45&&50 &$\frac{3}{7}$&0.43\\[1ex]
11 &$\mathbf{-\frac{1}{2}}$&\bf{-0.50}&&31 &$\mathbf{-\frac{17}{30}}$&\bf{--0.57}&&51 &$\mathbf{-\frac{29}{50}}$&\bf{--0.58}\\[1ex]
12 &$\frac{6}{11}$&0.55&&32 &$\frac{14}{31}$&0.45&&52 &$\frac{22}{51}$&0.43\\[1ex]
13 &$\mathbf{-\frac{1}{2}}$&\bf{-0.50}&&33 &$\mathbf{-\frac{9}{16}}$&\bf{--0.56}&&53 &$\mathbf{-\frac{15}{26}}$&\bf{--0.58}\\[1ex]
14 &$\frac{7}{13}$&0.54&&34 &$\frac{5}{11}$&0.45&&54 &$\frac{23}{53}$&0.43\\[1ex]
15 &$\frac{1}{2}$&0.50&&35 &$\frac{15}{34}$&0.44&&55 &$\frac{23}{54}$&0.43\\[1ex]
16 &$\mathbf{-\frac{8}{15}}$&\bf{-0.53}&&36 &$\mathbf{-\frac{4}{7}}$&\bf{--0.57}&&56 &$\mathbf{-\frac{32}{55}}$&\bf{--0.58}\\[1ex]
17 &$\frac{1}{2}$&0.50&&37 &$\frac{4}{9}$&0.44&&57 &$\frac{3}{7}$&0.43\\[1ex]
18 &$\mathbf{-\frac{9}{17}}$&\bf{-0.53}&&38 &$\mathbf{-\frac{21}{37}}$&\bf{--0.57}&&58 &$\mathbf{-\frac{11}{19}}$&\bf{--0.58}\\[1ex]
19 &$\frac{1}{2}$&0.50&&39 &$\frac{17}{38}$&0.45&&59 &$\frac{25}{58}$&0.43\\[1ex]
20 &$\frac{9}{19}$&0.47&&40 &$\frac{17}{39}$&0.44&&60 &$\frac{25}{59}$&0.42\\[1ex]
\cline{1-3}
\cline{5-7}
\cline{9-11}
\end{tabular}
\caption{Table with exact and approximate (rounded) values of $f(z)$ for $\rho = 1.6$ and $z \in \{2, 3, \dots, 60\}.$}
\label{fig-table}
\end{figure}

\clearpage

\subsection{Brick Surplus}
\label{apx-brick-surplus}

See Figure~\ref{fig-surplus-integer} for values of normalized brick
surplus at integer values of $\lambda$ and
Figure~\ref{fig-surplus-float} for values at non-integer breakpoint
values of $\lambda$.
Note that the normalized brick surplus sometimes decreases.

\begin{figure}[h]
    \centering
    \begin{tabular}{|c|c|c|c|c|c|c|c|}
        \cline{1-2}
        \cline{4-5}
        \cline{7-8}
        $\lambda$ & surplus & \hspace{1.0cm} &$\lambda$ & surplus & \hspace{1.0cm} &$\lambda$ & surplus\\
        \cline{1-2}
        \cline{4-5}
        \cline{7-8}
        1 & 0.000 && 21 & 0.321 && 41 & 0.561\\
        2 & 0.000 && 22 & 0.366 && 42 & 0.576\\
        3 & 0.000 && 23 & 0.377 && 43 & 0.576\\
        4 & 0.083 && 24 & 0.377 && 44 & 0.599\\
        5 & 0.083 && 25 & 0.417 && 45 & 0.615\\
        6 & 0.083 && 26 & 0.405 && 46 & 0.615\\
        7 & 0.133 && 27 & 0.405 && 47 & 0.636\\
        8 & 0.133 && 28 & 0.405 && 48 & 0.638\\
        9 & 0.244 && 29 & 0.406 && 49 & 0.658\\
        10 & 0.253 && 30 & 0.440 && 50 & 0.690\\
        11 & 0.253 && 31 & 0.457 && 51 & 0.690\\
        12 & 0.297 && 32 & 0.457 && 52 & 0.709\\
        13 & 0.220 && 33 & 0.457 && 53 & 0.710\\
        14 & 0.220 && 34 & 0.472 && 54 & 0.710\\
        15 & 0.252 && 35 & 0.500 && 55 & 0.728\\
        16 & 0.252 && 36 & 0.528 && 56 & 0.732\\
        17 & 0.310 && 37 & 0.539 && 57 & 0.749\\
        18 & 0.276 && 38 & 0.539 && 58 & 0.765\\
        19 & 0.276 && 39 & 0.561 && 59 & 0.765\\
        20 & 0.321 && 40 & 0.561 && 60 & 0.782\\
        \cline{1-2}
        \cline{4-5}
        \cline{7-8}
    \end{tabular}
    \caption{Table of approximate values of the normalized brick surplus for integer values of $\lambda$.}
    \label{fig-surplus-integer}
\end{figure}

\begin{figure}[!h]
    \centering
    \begin{tabular}{|c|c|c|c|c|c|c|}
        \cline{1-3}
        \cline{5-7}
        bag cost & $\lambda$ & surplus & \hspace{1cm} &
        bag cost & $\lambda$ & surplus \\
        \cline{1-3}
        \cline{5-7}
        1 & 3.667 & 0.167 &&
        12 & 33.490 & 0.472 \\
        2 & 6.350 & 0.133 &&
        13 & 36.206 & 0.539 \\
        3 & 9.044 & 0.253 &&
        14 & 38.923 & 0.563 \\
        4 & 11.761 & 0.317 &&
        15 & 41.642 & 0.576 \\
        5 & 14.477 & 0.252 &&
        16 & 44.360 & 0.615 \\
        6 & 17.188 & 0.321 &&
        17 & 47.076 & 0.638 \\
        7 & 19.902 & 0.321 &&
        18 & 49.795 & 0.690 \\
        8 & 22.622 & 0.393 &&
        19 & 52.514 & 0.719 \\
        9 & 25.338 & 0.430 &&
        20 & 55.231 & 0.732 \\
        10 & 28.052 & 0.406 &&
        21 & 57.948 & 0.766 \\
        11 & 30.772 & 0.465 &&
        22 & 60.668 & 0.793 \\
        \cline{1-3}
        \cline{5-7}
    \end{tabular}
    \caption{Table of approximate values of the normalized brick surplus at the points when we stop using bags of given cost.}
    \label{fig-surplus-float}
\end{figure}

\clearpage

\subsection{Small cases for bricks}
\label{apx-small-m}

Here we give a computer program that verifies that algorithm~\nameref{alg-bag-size} produces bags of total size at least $n$ for all instances with  $m \le 144$ and $\frac n m \le 60$.

We provide an implementation in \texttt{Python 3.11}.
We represent numbers as fractions to avoid issues with numerical precision.
We implement the equivalent Algorithm \nameref{alg-integral-fractional}, since it is faster than Algorithm \nameref{alg-bag-size}.
The code should take at most a few minutes to run on modern hardware.

\begin{code}
from fractions import Fraction
from math import ceil, floor

def bricks_alt(
    n: int,
    m: int,
) -> dict[int, Fraction]:
    """
    n: number of jobs
    m: number of machines

    returns: Solution in format {bag_cost: number_of_bags}
    """
    c = n
    r = m
    I = {}
    while r > 0 and c > 0:
        z = ceil(Fraction(c, m))
        x = min(r, ceil(Fraction(c - m * (z - 1), z)))
        c -= x * z
        r -= x
        I[z] = x

    return I

rho = Fraction(8, 5)
for m in range(1, 145):
    for n in range(1, 60 * m + 1):
        I = bricks_alt(n, m)
        solution_size = sum(
            floor(bag_cost * rho) * bag_count
            for bag_cost, bag_count in I.items()
        )
        if solution_size < n:
            raise Exception(
                f"Failed for: {n=} {m=} {rho=} \n Solution: {I}"
            )
\end{code}